\documentclass[letterpaper, 10 pt, conference]{ieeeconf}  

\IEEEoverridecommandlockouts                              

\usepackage{epsfig} 
\usepackage{times} 
\usepackage{amsmath} 
\usepackage{amssymb}  
\usepackage[normalem]{ulem}
\usepackage{booktabs}
\usepackage{graphicx} 
\graphicspath{{figures/}}
\usepackage{tikz}
\usetikzlibrary{backgrounds,angles,quotes,calc,patterns,decorations.pathmorphing,decorations.markings,positioning}
\usepackage{pgfplots}
\usetikzlibrary{intersections}
\usepackage{bm}
\usetikzlibrary{positioning}
\usepgfplotslibrary{fillbetween}
\pgfplotsset{compat=1.5}
\usetikzlibrary{pgfplots.groupplots}
\usetikzlibrary{shapes.geometric,backgrounds}

\usepackage{tikz}
\usepackage{xifthen}

\pgfdeclarelayer{bg}    
\pgfsetlayers{bg,main}  
\definecolor{beige}{RGB}{245, 245, 220}

\definecolor{darkgrey}{RGB}{75, 75, 75}
\definecolor{lightgrey}{RGB}{250, 250, 250}

\usetikzlibrary{calc, arrows, fit, positioning, patterns, decorations.pathreplacing, shapes}
\tikzstyle{dash} = [dashed, -latex,>=latex]
\tikzstyle{line} = [draw, -latex,>=latex]
\tikzstyle{smallbox} = [draw, minimum size=5.0mm]
\tikzstyle{box} = [draw, minimum size=7.0mm]
\tikzstyle{bigbox} = [draw, minimum size=10.0mm]
\tikzstyle{rectangle} = [draw, minimum width=10.0mm, minimum height=20.0mm]
\tikzstyle{switch} = [trapezium, trapezium angle=120, draw, rotate=90,  inner ysep=5pt, outer sep=5pt,
minimum height=7mm, minimum width=7mm]
\tikzstyle{roundbox} = [draw, circle, inner sep=0pt, minimum size=3mm]
\tikzstyle{clamped} = [draw, fill=darkgrey, minimum size=0.15cm]
\tikzstyle{msgcircle} = [shape=circle, draw, inner sep=0pt, minimum size=4mm, fill=white, font=\scriptsize]
\tikzstyle{darkmsgcircle} = [shape=circle, draw, inner sep=0pt, minimum size=4mm, fill=darkgrey, text=white, font=\scriptsize]
\tikzstyle{redmsgcircle} = [shape=circle, draw=red, inner sep=0pt, minimum size=4mm, text=red, font=\scriptsize]
\tikzstyle{reddarkmsgcircle} = [shape=circle, draw=red, inner sep=0pt, minimum size=4mm, fill=red, text=white, font=\scriptsize]
\tikzstyle{msgdoublecircle} = [shape=circle, double, double distance=1.5pt, draw, inner sep=0pt, minimum size=5mm, fill=white]
\tikzstyle{darkmsgdoublecircle} = [shape=circle, double, double distance=1.5pt, draw, inner sep=0pt, minimum size=5mm, fill=darkgrey, text=white, font=\bfseries]

\newcommand{\msg}[6]{
      \ifthenelse{\isin{#1}{left} \AND \isin{#2}{down}}{
            \coordinate (anchor) at ($({#3})!{#5}!({#4})$);
            \node[xshift=-6.0mm] at (anchor) {#6};
            \node[xshift=-1.0mm] at (anchor) {$\downarrow$};
      }{}
      \ifthenelse{\isin{#1}{right} \AND \isin{#2}{down}}{
            \coordinate (anchor) at ($({#3})!{#5}!({#4})$);
            \node[xshift=6.0mm] at (anchor) {#6};
            \node[xshift=1.0mm] at (anchor) {$\downarrow$};
      }{}

      \ifthenelse{\isin{#1}{down} \AND \isin{#2}{right}}{
            \coordinate (anchor) at ($({#3})!{#5}!({#4})$);
            \node[ yshift=-4.0mm] at (anchor) {#6};
            \node[yshift=-1.0mm] at (anchor) {$\rightarrow$};
      }{}
      \ifthenelse{\isin{#1}{up} \AND \isin{#2}{right}}{
            \coordinate (anchor) at ($({#3})!{#5}!({#4})$);
            \node[ yshift=4.0mm] at (anchor) {#6};
            \node[yshift=1.0mm] at (anchor) {$\rightarrow$};
      }{}

      \ifthenelse{\isin{#1}{down} \AND \isin{#2}{left}}{
            \coordinate (anchor) at ($({#3})!{#5}!({#4})$);
            \node[ yshift=-4.0mm] at (anchor) {#6};
            \node[yshift=-1.0mm] at (anchor) {$\leftarrow$};
      }{}
      \ifthenelse{\isin{#1}{up} \AND \isin{#2}{left}}{
            \coordinate (anchor) at ($({#3})!{#5}!({#4})$);
            \node[ yshift=4.0mm] at (anchor) {#6};
            \node[yshift=1.0mm] at (anchor) {$\leftarrow$};
      }{}

      \ifthenelse{\isin{#1}{left} \AND \isin{#2}{up}}{
            \coordinate (anchor) at ($({#3})!{#5}!({#4})$);
            \node[ xshift=-6.0mm] at (anchor) {#6};
            \node[xshift=-1.0mm] at (anchor) {$\uparrow$};
      }{}
      \ifthenelse{\isin{#1}{right} \AND \isin{#2}{up}}{
            \coordinate (anchor) at ($({#3})!{#5}!({#4})$);
            \node[ xshift=6.0mm] at (anchor) {#6};
            \node[xshift=1.0mm] at (anchor) {$\uparrow$};
      }{}
}

\newcommand{\msgcircle}[6]{
      \ifthenelse{\isin{#1}{left} \AND \isin{#2}{down}}{
            \coordinate (anchor) at ($({#3})!{#5}!({#4})$);
            \node[msgcircle,xshift=-5.0mm] at (anchor) {#6};
            \node[xshift=-1.5mm] at (anchor) {$\downarrow$};
      }{}
      \ifthenelse{\isin{#1}{right} \AND \isin{#2}{down}}{
            \coordinate (anchor) at ($({#3})!{#5}!({#4})$);
            \node[msgcircle,xshift=5.0mm] at (anchor) {#6};
            \node[xshift=1.5mm] at (anchor) {$\downarrow$};
      }{}

      \ifthenelse{\isin{#1}{down} \AND \isin{#2}{right}}{
            \coordinate (anchor) at ($({#3})!{#5}!({#4})$);
            \node[msgcircle, yshift=-5.0mm] at (anchor) {#6};
            \node[yshift=-2.0mm] at (anchor) {$\rightarrow$};
      }{}
      \ifthenelse{\isin{#1}{up} \AND \isin{#2}{right}}{
            \coordinate (anchor) at ($({#3})!{#5}!({#4})$);
            \node[msgcircle, yshift=5.0mm] at (anchor) {#6};
            \node[yshift=2.0mm] at (anchor) {$\rightarrow$};
      }{}

      \ifthenelse{\isin{#1}{down} \AND \isin{#2}{left}}{
            \coordinate (anchor) at ($({#3})!{#5}!({#4})$);
            \node[msgcircle, yshift=-5.0mm] at (anchor) {#6};
            \node[yshift=-2.0mm] at (anchor) {$\leftarrow$};
      }{}
      \ifthenelse{\isin{#1}{up} \AND \isin{#2}{left}}{
            \coordinate (anchor) at ($({#3})!{#5}!({#4})$);
            \node[msgcircle, yshift=5.0mm] at (anchor) {#6};
            \node[yshift=2.0mm] at (anchor) {$\leftarrow$};
      }{}

      \ifthenelse{\isin{#1}{left} \AND \isin{#2}{up}}{
            \coordinate (anchor) at ($({#3})!{#5}!({#4})$);
            \node[msgcircle, xshift=-5.0mm] at (anchor) {#6};
            \node[xshift=-1.5mm] at (anchor) {$\uparrow$};
      }{}
      \ifthenelse{\isin{#1}{right} \AND \isin{#2}{up}}{
            \coordinate (anchor) at ($({#3})!{#5}!({#4})$);
            \node[msgcircle, xshift=5.0mm] at (anchor) {#6};
            \node[xshift=1.5mm] at (anchor) {$\uparrow$};
      }{}
}

\newcommand{\darkmsg}[6]{
      \ifthenelse{\isin{#1}{left} \AND \isin{#2}{down}}{
            \coordinate (anchor) at ($({#3})!{#5}!({#4})$);
            \node[darkmsgcircle, xshift=-5mm] at (anchor) {#6};
            \node[xshift=-1.5mm] at (anchor) {$\downarrow$};
      }{}
      \ifthenelse{\isin{#1}{right} \AND \isin{#2}{down}}{
            \coordinate (anchor) at ($({#3})!{#5}!({#4})$);
            \node[darkmsgcircle, xshift=5mm] at (anchor) {#6};
            \node[xshift=1.5mm] at (anchor) {$\downarrow$};
      }{}

      \ifthenelse{\isin{#1}{down} \AND \isin{#2}{right}}{
            \coordinate (anchor) at ($({#3})!{#5}!({#4})$);
            \node[darkmsgcircle, yshift=-5.0mm] at (anchor) {#6};
            \node[yshift=-2.0mm] at (anchor) {$\rightarrow$};
      }{}
      \ifthenelse{\isin{#1}{up} \AND \isin{#2}{right}}{
            \coordinate (anchor) at ($({#3})!{#5}!({#4})$);
            \node[darkmsgcircle, yshift=5.0mm] at (anchor) {#6};
            \node[yshift=2.0mm] at (anchor) {$\rightarrow$};
      }{}

      \ifthenelse{\isin{#1}{down} \AND \isin{#2}{left}}{
            \coordinate (anchor) at ($({#3})!{#5}!({#4})$);
            \node[darkmsgcircle, yshift=-5.0mm] at (anchor) {#6};
            \node[yshift=-2.0mm] at (anchor) {$\leftarrow$};
      }{}
      \ifthenelse{\isin{#1}{up} \AND \isin{#2}{left}}{
            \coordinate (anchor) at ($({#3})!{#5}!({#4})$);
            \node[darkmsgcircle, yshift=5.0mm] at (anchor) {#6};
            \node[yshift=2.0mm] at (anchor) {$\leftarrow$};
      }{}

      \ifthenelse{\isin{#1}{left} \AND \isin{#2}{up}}{
            \coordinate (anchor) at ($({#3})!{#5}!({#4})$);
            \node[darkmsgcircle, xshift=-5.0mm] at (anchor) {#6};
            \node[xshift=-1.5mm] at (anchor) {$\uparrow$};
      }{}
      \ifthenelse{\isin{#1}{right} \AND \isin{#2}{up}}{
            \coordinate (anchor) at ($({#3})!{#5}!({#4})$);
            \node[darkmsgcircle, xshift=5.0mm] at (anchor) {#6};
            \node[xshift=1.5mm] at (anchor) {$\uparrow$};
      }{}
}

\newcommand{\redbackmsg}[6]{
      \ifthenelse{\isin{#1}{left} \AND \isin{#2}{down}}{
            \coordinate (anchor) at ($({#3})!{#5}!({#4})$);
            \node[reddarkmsgcircle, xshift=-5mm] at (anchor) {#6};
            \node[xshift=-1.5mm] at (anchor) {$\downarrow$};
      }{}
      \ifthenelse{\isin{#1}{right} \AND \isin{#2}{down}}{
            \coordinate (anchor) at ($({#3})!{#5}!({#4})$);
            \node[reddarkmsgcircle, xshift=5mm] at (anchor) {#6};
            \node[xshift=1.5mm] at (anchor) {$\downarrow$};
      }{}

      \ifthenelse{\isin{#1}{down} \AND \isin{#2}{right}}{
            \coordinate (anchor) at ($({#3})!{#5}!({#4})$);
            \node[reddarkmsgcircle, yshift=-5.0mm] at (anchor) {#6};
            \node[yshift=-2.0mm] at (anchor) {$\rightarrow$};
      }{}
      \ifthenelse{\isin{#1}{up} \AND \isin{#2}{right}}{
            \coordinate (anchor) at ($({#3})!{#5}!({#4})$);
            \node[reddarkmsgcircle, yshift=5.0mm] at (anchor) {#6};
            \node[yshift=2.0mm] at (anchor) {$\rightarrow$};
      }{}

      \ifthenelse{\isin{#1}{down} \AND \isin{#2}{left}}{
            \coordinate (anchor) at ($({#3})!{#5}!({#4})$);
            \node[reddarkmsgcircle, yshift=-5.0mm] at (anchor) {#6};
            \node[yshift=-2.0mm] at (anchor) {$\leftarrow$};
      }{}
      \ifthenelse{\isin{#1}{up} \AND \isin{#2}{left}}{
            \coordinate (anchor) at ($({#3})!{#5}!({#4})$);
            \node[reddarkmsgcircle, yshift=5.0mm] at (anchor) {#6};
            \node[yshift=2.0mm] at (anchor) {$\leftarrow$};
      }{}

      \ifthenelse{\isin{#1}{left} \AND \isin{#2}{up}}{
            \coordinate (anchor) at ($({#3})!{#5}!({#4})$);
            \node[reddarkmsgcircle, xshift=-5.0mm] at (anchor) {#6};
            \node[xshift=-1.5mm] at (anchor) {$\uparrow$};
      }{}
      \ifthenelse{\isin{#1}{right} \AND \isin{#2}{up}}{
            \coordinate (anchor) at ($({#3})!{#5}!({#4})$);
            \node[reddarkmsgcircle, xshift=5.0mm] at (anchor) {#6};
            \node[xshift=1.5mm] at (anchor) {$\uparrow$};
      }{}
}

\newcommand{\redmsg}[6]{
      \ifthenelse{\isin{#1}{left} \AND \isin{#2}{down}}{
            \coordinate (anchor) at ($({#3})!{#5}!({#4})$);
            \node[redmsgcircle, xshift=-5mm] at (anchor) {#6};
            \node[xshift=-1.5mm] at (anchor) {$\downarrow$};
      }{}
      \ifthenelse{\isin{#1}{right} \AND \isin{#2}{down}}{
            \coordinate (anchor) at ($({#3})!{#5}!({#4})$);
            \node[redmsgcircle, xshift=5mm] at (anchor) {#6};
            \node[xshift=1.5mm] at (anchor) {$\downarrow$};
      }{}

      \ifthenelse{\isin{#1}{down} \AND \isin{#2}{right}}{
            \coordinate (anchor) at ($({#3})!{#5}!({#4})$);
            \node[redmsgcircle, yshift=-5.0mm] at (anchor) {#6};
            \node[yshift=-2.0mm] at (anchor) {$\rightarrow$};
      }{}
      \ifthenelse{\isin{#1}{up} \AND \isin{#2}{right}}{
            \coordinate (anchor) at ($({#3})!{#5}!({#4})$);
            \node[redmsgcircle, yshift=5.0mm] at (anchor) {#6};
            \node[yshift=2.0mm] at (anchor) {$\rightarrow$};
      }{}

      \ifthenelse{\isin{#1}{down} \AND \isin{#2}{left}}{
            \coordinate (anchor) at ($({#3})!{#5}!({#4})$);
            \node[redmsgcircle, yshift=-5.0mm] at (anchor) {#6};
            \node[yshift=-2.0mm] at (anchor) {$\leftarrow$};
      }{}
      \ifthenelse{\isin{#1}{up} \AND \isin{#2}{left}}{
            \coordinate (anchor) at ($({#3})!{#5}!({#4})$);
            \node[redmsgcircle, yshift=5.0mm] at (anchor) {#6};
            \node[yshift=2.0mm] at (anchor) {$\leftarrow$};
      }{}

      \ifthenelse{\isin{#1}{left} \AND \isin{#2}{up}}{
            \coordinate (anchor) at ($({#3})!{#5}!({#4})$);
            \node[redmsgcircle, xshift=-5.0mm] at (anchor) {#6};
            \node[xshift=-1.5mm] at (anchor) {$\uparrow$};
      }{}
      \ifthenelse{\isin{#1}{right} \AND \isin{#2}{up}}{
            \coordinate (anchor) at ($({#3})!{#5}!({#4})$);
            \node[redmsgcircle, xshift=5.0mm] at (anchor) {#6};
            \node[xshift=1.5mm] at (anchor) {$\uparrow$};
      }{}
}

\newcommand{\bwmsg}[6]{
      \ifthenelse{\isin{#1}{left} \AND \isin{#2}{down}}{
            \coordinate (anchor) at ($({#3})!{#5}!({#4})$);
            \node[msgdoublecircle, xshift=-5.5mm] at (anchor) {#6};
            \node[xshift=-1.5mm] at (anchor) {$\downarrow$};
      }{}
      \ifthenelse{\isin{#1}{right} \AND \isin{#2}{down}}{
            \coordinate (anchor) at ($({#3})!{#5}!({#4})$);
            \node[msgdoublecircle, xshift=5.5mm] at (anchor) {#6};
            \node[xshift=1.5mm] at (anchor) {$\downarrow$};
      }{}

      \ifthenelse{\isin{#1}{down} \AND \isin{#2}{right}}{
            \coordinate (anchor) at ($({#3})!{#5}!({#4})$);
            \node[msgdoublecircle, yshift=-6.0mm] at (anchor) {#6};
            \node[yshift=-2.0mm] at (anchor) {$\rightarrow$};
      }{}
      \ifthenelse{\isin{#1}{up} \AND \isin{#2}{right}}{
            \coordinate (anchor) at ($({#3})!{#5}!({#4})$);
            \node[msgdoublecircle, yshift=6.0mm] at (anchor) {#6};
            \node[yshift=2.0mm] at (anchor) {$\rightarrow$};
      }{}

      \ifthenelse{\isin{#1}{down} \AND \isin{#2}{left}}{
            \coordinate (anchor) at ($({#3})!{#5}!({#4})$);
            \node[msgdoublecircle, yshift=-6.0mm] at (anchor) {#6};
            \node[yshift=-2.0mm] at (anchor) {$\leftarrow$};
      }{}
      \ifthenelse{\isin{#1}{up} \AND \isin{#2}{left}}{
            \coordinate (anchor) at ($({#3})!{#5}!({#4})$);
            \node[msgdoublecircle, yshift=6.0mm] at (anchor) {#6};
            \node[yshift=2.0mm] at (anchor) {$\leftarrow$};
      }{}

      \ifthenelse{\isin{#1}{left} \AND \isin{#2}{up}}{
            \coordinate (anchor) at ($({#3})!{#5}!({#4})$);
            \node[msgdoublecircle, xshift=-5.5mm] at (anchor) {#6};
            \node[xshift=-1.5mm] at (anchor) {$\uparrow$};
      }{}
      \ifthenelse{\isin{#1}{right} \AND \isin{#2}{up}}{
            \coordinate (anchor) at ($({#3})!{#5}!({#4})$);
            \node[msgdoublecircle, xshift=5.5mm] at (anchor) {#6};
            \node[xshift=1.5mm] at (anchor) {$\uparrow$};
      }{}
}

\newcommand{\bwdarkmsg}[6]{
      \ifthenelse{\isin{#1}{left} \AND \isin{#2}{down}}{
            \coordinate (anchor) at ($({#3})!{#5}!({#4})$);
            \node[darkmsgdoublecircle, xshift=-5.5mm] at (anchor) {#6};
            \node[xshift=-1.5mm] at (anchor) {$\downarrow$};
      }{}
      \ifthenelse{\isin{#1}{right} \AND \isin{#2}{down}}{
            \coordinate (anchor) at ($({#3})!{#5}!({#4})$);
            \node[darkmsgdoublecircle, xshift=5.5mm] at (anchor) {#6};
            \node[xshift=1.5mm] at (anchor) {$\downarrow$};
      }{}

      \ifthenelse{\isin{#1}{down} \AND \isin{#2}{right}}{
            \coordinate (anchor) at ($({#3})!{#5}!({#4})$);
            \node[darkmsgdoublecircle, yshift=-6.0mm] at (anchor) {#6};
            \node[yshift=-2.0mm] at (anchor) {$\rightarrow$};
      }{}
      \ifthenelse{\isin{#1}{up} \AND \isin{#2}{right}}{
            \coordinate (anchor) at ($({#3})!{#5}!({#4})$);
            \node[darkmsgdoublecircle, yshift=6.0mm] at (anchor) {#6};
            \node[yshift=2.0mm] at (anchor) {$\rightarrow$};
      }{}

      \ifthenelse{\isin{#1}{down} \AND \isin{#2}{left}}{
            \coordinate (anchor) at ($({#3})!{#5}!({#4})$);
            \node[darkmsgdoublecircle, yshift=-6.0mm] at (anchor) {#6};
            \node[yshift=-2.0mm] at (anchor) {$\leftarrow$};
      }{}
      \ifthenelse{\isin{#1}{up} \AND \isin{#2}{left}}{
            \coordinate (anchor) at ($({#3})!{#5}!({#4})$);
            \node[darkmsgdoublecircle, yshift=6.0mm] at (anchor) {#6};
            \node[yshift=2.0mm] at (anchor) {$\leftarrow$};
      }{}

      \ifthenelse{\isin{#1}{left} \AND \isin{#2}{up}}{
            \coordinate (anchor) at ($({#3})!{#5}!({#4})$);
            \node[darkmsgdoublecircle, xshift=-5.5mm] at (anchor) {#6};
            \node[xshift=-1.5mm] at (anchor) {$\uparrow$};
      }{}
      \ifthenelse{\isin{#1}{right} \AND \isin{#2}{up}}{
            \coordinate (anchor) at ($({#3})!{#5}!({#4})$);
            \node[darkmsgdoublecircle, xshift=5.5mm] at (anchor) {#6};
            \node[xshift=1.5mm] at (anchor) {$\uparrow$};
      }{}
}

\tikzset{mainstyle/.style={fill=white, draw=black, shape=rectangle, align=center}}

\tikzset{dstyle/.style={mainstyle, minimum size=4mm, inner sep=0pt, text width=4mm}}

\tikzset{sstyle/.style={mainstyle, minimum size=5mm, inner sep=0pt, text width=5mm}}

\tikzset{ostyle/.style={fill=darkgrey, draw=black, shape=rectangle, minimum size=0.2cm, inner sep=0pt, text width=2mm}}

\tikzstyle{observation}=[ostyle]
\tikzstyle{deterministic}=[dstyle]
\tikzstyle{stochastic}=[sstyle]

\tikzstyle{filter}=[mainstyle, minimum width=1cm, minimum height=0.5cm]
\tikzstyle{selector}=[fill=white, draw=black, shape=trapezium, rotate=180, minimum width=1cm, minimum height=0.5cm]

\usepackage[disable]{todonotes}
 
\newcommand{\wouter}[2][] {\todo[inline,backgroundcolor=red!20!white, #1]{(Wouter) #2}}

\usepackage{hyperref}

\def\-{\text{-}}
\def\+{\text{+}}
\def\tm{\! - \!}
\def\tp{\! + \!}
\def\tr{\text{tr}}
\newcommand\given[1][]{\:#1\vert\:}

\makeatletter
\newcommand{\smallbullet}{} 
\DeclareRobustCommand\smallbullet{%
  \mathord{\mathpalette\smallbullet@{0.5}}%
}
\newcommand{\smallbullet@}[2]{%
  \vcenter{\hbox{\scalebox{#2}{$\m@th#1\bullet$}}}%
}
\makeatother

\newcommand*\wcircled[1]{\tikz[baseline=(char.base)]{
            \node[shape=circle,draw,minimum size=4mm,inner sep=0pt] (char) {#1};}}

\newtheorem{theorem}{Theorem}
\newtheorem{lemma}{Lemma}
\usepackage{subfigure}  

\title{\LARGE \bf
Online Bayesian system identification\\ in multivariate autoregressive models via message passing
}

\author{Tim N. Nisslbeck$^{1}$ and Wouter M. Kouw$^{1}$
\thanks{*Supported by the Eindhoven Artificial Intelligence Systems Institute.}
\thanks{$^{1}$Nisslbeck and Kouw are with the Department of Electrical Engineering, Eindhoven University of Technology, Eindhoven, the Netherlands. Corresponding email: {\tt\small t.n.nisslbeck@tue.nl}}%
}

\begin{document}

\maketitle
\thispagestyle{empty}
\pagestyle{empty}

\begin{abstract}
We propose a recursive Bayesian estimation procedure for multivariate autoregressive models with exogenous inputs based on message passing in a factor graph. Unlike recursive least-squares, our method produces full posterior distributions for both the autoregressive coefficients and noise precision. The uncertainties regarding these estimates propagate into the uncertainties on predictions for future system outputs, and support online model evidence calculations. We demonstrate convergence empirically on a synthetic autoregressive system and competitive performance on a double mass-spring-damper system.
\end{abstract}

\section{INTRODUCTION} \label{sec:introduction}
Autoregressive models capture dynamical systems  with simple yet expressive structures \cite{tiao1964bayesian,hannan1989recursive,karlsson2013forecasting,nisslbeck2024coupled}.
In multivariate autoregressive models with exogenous inputs (MARX), the evolution of the signal incorporates past observations and controls, producing substantial uncertainty during parameter estimation.
Bayesian inference procedures can quantify this uncertainty and propagate it towards future predictions \cite{penny2007multivariate,shaarawy2008bayesian}.
Quantified uncertainty is valuable on its own, but also useful to sensor fusion, optimal experimental design and adaptive control \cite{castaldo2014multi,chaloner1995bayesian,ta2014factor,hoffmann2017linear,palmieri2022unifying}.
We present an exact recursive Bayesian estimator whose computation is distributed over a probabilistic graphical model.

Bayesian inference in multivariate autoregressive models has a rich history, especially in econometrics \cite{tiao1964bayesian,karlsson2013forecasting}. Scientists typically employ Markov Chain Monte Carlo methods to obtain approximate posterior distributions, but such techniques are too computationally expensive for online system identification or adaptive control systems. Recursive estimators are more suited, but generally lack posterior uncertainty on parameters or produce approximate posterior distributions \cite{hannan1989recursive,penny2007multivariate}. 
We propose an exact recursive Bayesian estimator using the matrix normal Wishart distribution and cast the inference procedure as a message passing algorithm on a factor graph. Factor graphs have several key features. Firstly, they provide a more accessible visual representation of a probabilistic model and the structure of incoming data streams \cite{loeliger2007factor,palmieri2022unifying}. Secondly, by assigning computation to nodes, model design becomes more modular, allowing inference to be automated and distributed over devices \cite{podusenko2021message,bagaev2023reactive}. Thirdly, because messages propagate uncertainty along the graph, the final output uncertainty can be tracked back to individual sources of uncertainty, which allows a prediction to be explained in terms of, for example, volatility versus parameter uncertainty \cite{lecue2020role}. Lastly, message passing unifies a variety of algorithms, from signal filtering to optimal control and path planning \cite{podusenko2021message,hoffmann2017linear,palmieri2022unifying}.


Our contributions are computational in nature:
\begin{itemize}
    \item A message passing algorithm is derived for recursive Bayesian inference in MARX models.
    \item A distribution is derived for predicting future system outputs that incorporates parameter uncertainty. 
\end{itemize}
The performance of our proposed procedure is empirically evaluated in two experimental settings.

\section{PROBLEM STATEMENT} \label{sec:problem}
We study discrete-time state-space systems that evolve over time according to a state transition function $f: \mathbb{R}^{D_z} \times \mathbb{R}^{D_u} \mapsto \mathbb{R}^{D_z}$ and outputs noisy measurements $y_t \in \mathbb{R}^{D_y}$ through a measurement function $g: \mathbb{R}^{D_z} \mapsto \mathbb{R}^{D_y}$:
\begin{equation} \label{eq:problem:state-transition}
z_t = f(z_{t\-1}, u_{t}) \, , \qquad 
y_t = g(z_t) + e_t \, ,
\end{equation}
with stochastic disturbance $e_t \in \mathbb{R}^{D_y}$ and the goal to predict future outputs $y_{\tau}$ for $\tau > t$ given future inputs $u_{\tau}$ using a model of the system's dynamics.

\section{MODEL SPECIFICATION} \label{sec:model-specification}
We consider an order-$N$ MARX model. Let $y_t \in \mathbb{R}^{D_y}$ be a multivariate signal and let 
\begin{align}
\bar{y}_{t\-1}
&\triangleq
\begin{bmatrix}
    y_{t-1,1} & y_{t-2,1} & \dots & y_{t-N_y,1} \\
    \vdots & \dots & \dots & \vdots \\
    y_{t-1,D_y} & y_{t-2,D_y} & \dots & y_{t-N_y,D_y}
\end{bmatrix} \, ,
\end{align}
be the observation history of memory size $N_y$.
Similarly, we have a matrix of previous values of exogeneous (control) signals
\begin{align}
\bar{u}_{t}
&\triangleq
\begin{bmatrix}
    u_{t,1} & u_{t-1,1} & \dots & u_{t-N_u+1,1} \\
    \vdots & \dots & \dots & \vdots \\
    u_{t,D_u} & u_{t-1,D_u} & \dots & u_{t-N_u+1,D_u}
\end{bmatrix}
\end{align}
with control memory size $N_u$.
We shall reshape both matrices $\bar{y}_{t\-1}$ and $\bar{u}_{t\-1}$ into a $D_x =  N_y \times D_y + N_u \times D_u$ vector:
\begin{align} \label{eq:cat}
    x_t \triangleq \begin{bmatrix} \text{vec}(\bar{y}_{t\-1}) \  \text{vec}(\bar{u}_t) \end{bmatrix}^\intercal \, .
\end{align}
Our likelihood model is
\begin{align} \label{eq:param:likelihood}
p(&y_t \given A, W, x_t) = \mathcal{N}(y_t \given A^{\intercal} x_t, W^{-1}) \\
&= \sqrt{\frac{|W|}{(2\pi)^{D_y}}} \exp\big( \tm \frac{1}{2}(y_t \tm A^{\intercal} x_t)^\intercal W(y_t \tm A^{\intercal} x_t) \big) .
\end{align}
There are two unknowns: a matrix of autoregression coefficients \mbox{$A \in \mathbb{R}^{D_x \times D_y}$} and a precision matrix \mbox{$W \in \mathbb{R}^{D_y \times D_y}$}.
For computational convenience (see Section \ref{sec:inference:param}), we specify our prior distribution over $(A, W)$ to be a matrix normal Wishart distribution \cite[ID: D175]{joram_soch_2024_10495684}: 
\begin{align} \label{eq:param:prior}
     p(A,W) &= \mathcal{MNW}(A, W \given M_0, \Lambda_0^{-1}, \Omega_0^{-1}, \nu_0) \\
    &= \sqrt{\frac{ |\Omega_0|^{D_y} |\Lambda_0|^{\nu_0}}{(2\pi)^{D_x D_y} 2^{\nu_0 D_y}}} \frac{\sqrt{|W|^{\nu_0 + D_x - D_y-1}}}{\Gamma_{D_y}(\nu_0 / 2)}  \\
    &\quad \exp\big(\tm \frac{1}{2} \tr\big[W \big( (A \tm M_0)^{\intercal} \Lambda_0 (A \tm M_0) \! + \! \Omega_0 \big) \big] \big) , \nonumber 
\end{align}
where $\Gamma_{D_y}(\cdot)$ is the $D_y$-dimensional multivariate gamma function. 
The coefficient matrix $A$ follows a matrix normal distribution with mean \mbox{$M_0 \in \mathbb{R}^{D_x \times D_y}$}, row covariance \mbox{$\Lambda_0^{-1} \in \mathbb{R}^{D_x \times D_x}$}, and column covariance $W^{-1}$,
\begin{align}
    &p(A \given W) = \mathcal{MN}(A \, | M_0, \Lambda_0^{-1}, W^{-1}) \label{eq:priors:A} \\
    &= \sqrt{\frac{|W|^{D_x} |\Lambda_0|^{D_y }}{(2\pi)^{D_x D_y}}} \exp \big(\tm \frac{1}{2}\tr\big[W (A \tm M_0)^{\intercal} \Lambda_0 (A \tm M_0 ) \big] \big) . \nonumber 
\end{align}
The precision matrix $W$ follows a Wishart distribution with scale matrix \mbox{$\Omega_0^{-1} \in \mathbb{R}^{D_y \times D_y}$} and degrees of freedom $\nu_0 \in \mathbb{R}$
\begin{align}
    p(W) &= \mathcal{W}(W \, | \, \Omega_0^{-1}, \nu_0) \label{eq:priors:W} \\
    &= \! \sqrt{\frac{|\Omega_0|^{\nu_0}}{2^{\nu_0 D_y}}} \frac{\sqrt{|W|^{\nu_0 \- D_y \- 1}}}{\Gamma_{D_y}(\frac{\nu_0}{2})} \exp\big(\tm \frac{1}{2}\tr\big[W \Omega_0\big] \big) .
\end{align}
Our goal is to infer a posterior distribution over parameters $A$ and $W$, and later to utilize these parameter posterior distributions to make predictions for future outputs $y_{\tau}$.

\subsection{Factor graph}
The probabilistic graphical model for the recursive model is straightforward, as it constitutes only a prior distribution and a likelihood function. We present a Forney-style factor graph in Fig.~\ref{fig:ffg}, where nodes correspond to factors, edges to variables and an edge may only be connected to two nodes \cite{loeliger2007factor}. In the graph, time flows from left to right, predictions from top to bottom and corrections from bottom to top. The factor node labeled $\mathcal{MNW}$ represents the matrix normal Wishart prediction distribution with its prior parameters. The dotted box represents the composite likelihood node, consisting of the concatenation operation described in \eqref{eq:cat}, the dot product operation between the matrix of autoregressive coefficients $A$ and the memory $x_t$, and the stochastic disturbance. The equality node connects the parameters $A,W$ to the likelihood nodes for each $t$.
\begin{figure}[thb]
    \centering
    \scalebox{1.2}{\begin{tikzpicture}
    \node [style=deterministic] (cat) {$\scriptstyle{[ \, ]}$};
    \node [style=deterministic, right=6mm of cat] (dot) {$\cdot$};
    \node [style=stochastic, right=14mm of dot] (N) {$\mathcal{N}$};
    \draw [-] (dot) -- node[below] {$\scriptstyle{(A^{\intercal}x_t,W)}$} (N);
    \draw [-] (cat) -- node[below] {$\scriptstyle{x_t}$} (dot);

    \node [style=observation, above of=cat, node distance=8mm] (us) {};
    \node [style=observation, below of=cat, node distance=8mm] (ys) {};
    \node [left of=us, node distance=4mm] (uss) {$\scriptstyle{\bar{u}_t}$};
    \node [left of=ys, node distance=5mm] (yss) {$\scriptstyle{\bar{y}_{t\text{-}1}}$};
    \draw[-] (us) -- (cat);
    \draw[-] (ys) -- (cat);

    \node [style=observation, below = 8mm of N] (y) {};
    \node [right of=y, node distance=4mm] (yk) {$\scriptstyle{y_t}$};
    \draw [-] (N) -- (y);
    
    \node [style=deterministic, above = 10mm of dot] (eqN) {$=$};
    \draw [-] (eqN) -- (dot);
    
    \node [left = 10mm of eqN] (prevzeta) {$\cdots$};
    \draw [-] (prevzeta) -- (eqN);
    \node [shape=rectangle, draw=black, inner ysep=2mm, minimum width=10mm, left=2mm of prevzeta] (zetaprior) {$\scriptstyle{\mathcal{MNW}}$};
    \draw [-] (zetaprior) -- (prevzeta);
    \node [style=observation, left = 2mm of zetaprior] (zetaprior_params) {};
    \node [left of=zetaprior_params, node distance=4mm] (zetaprior_params0) {$\begin{matrix} \scriptstyle{M_0} \\ \scriptstyle{\Lambda_0} \\ \scriptstyle{\Omega_0} \\ \scriptstyle{\nu_0} \end{matrix}$};
    \draw [-] (zetaprior_params) -- (zetaprior);

    \node [right = 14mm of eqN] (zetapost) {$\cdots$};
    \draw [-] (eqN) -- (zetapost);
    \node [above of=eqN, node distance=4mm] {$\scriptstyle{(A, W)}$};

    \node[dashed, fit=(cat)(N), draw, inner sep=2mm] (box) {};

    \msgcircle{up}{right}{prevzeta}{eqN}{0.5}{1};
    \msgcircle{right}{up}{dot}{eqN}{0.5}{2};
    \msgcircle{up}{right}{eqN}{zetapost}{0.5}{3};
    \msgcircle{left}{down}{N}{y}{0.7}{4};
    
\end{tikzpicture}}
    \caption{Forney-style factor graph of the MARX model in recursive form. A matrix normal Wishart node sends a prior message $1$ towards an equality node. A likelihood-based message $2$ passes upwards from the MARX likelihood node (dotted box), attached to observed variables $y_t$, $\bar{y}_{t\tm1}$ and $\bar{u}_t$. Combining the prior-based and likelihood-based message at the equality node yields the posterior (message $3$). Message 4 is the posterior predictive distribution for system output.}
    \label{fig:ffg}
\end{figure}
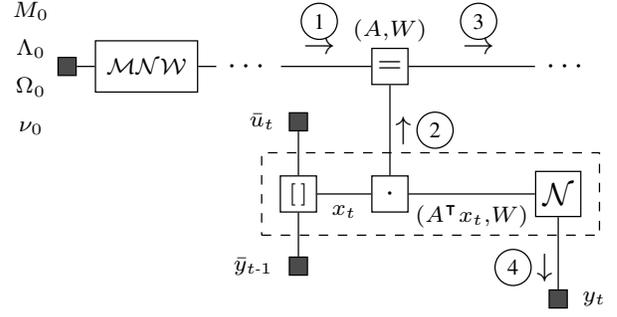

\section{INFERENCE} \label{sec:inference}
\subsection{Parameter estimation} \label{sec:inference:param}
We are interested in the posterior distribution, which we shall describe recursively:
\begin{align}
    p(A,W \given \mathcal{D}_{t}) = \frac{p(y_t \given A,W, x_t)}{p(y_t \given \mathcal{D}_t) } \ p(A,W \given \mathcal{D}_{t\-1}) \, ,
\end{align}
where the evidence term in the denominator is:
\begin{align} \label{eq:evidence}
& p(y_t | u_t, \mathcal{D}_{t\-1}) \! = \! \! \int \! \! p(y_t |  A, \! W, \! x_t) p(A, \! W |  \mathcal{D}_{t\-1}) \mathrm{d}(A,\! W) .
\end{align}

\begin{lemma} \label{eq:lemma1}
The MARX likelihood \eqref{eq:param:likelihood} combined with a matrix normal Wishart prior distribution over MARX coefficient matrix $A$ and precision matrix $W$ \eqref{eq:param:prior} yields a matrix normal Wishart distribution:
\begin{equation} \label{eq:param:posterior}
    p(A,W \given \mathcal{D}_t)
    = \mathcal{MNW}(A, W \given M_t, \Lambda_t^{-1}, \Omega_t^{-1}, \nu_t)
\end{equation}
with parameter update rules:
\begin{align}
    \nu_t
    &= \nu_{t\-1} + 1 , \\
    \Lambda_t
    &= \Lambda_{t\-1} + x_tx_t^\intercal , \\
    M_t
    &= (\Lambda_{t\-1} + x_tx_t^\intercal)^{-1}(\Lambda_{t\-1}M_{t\-1} + x_t y_t^\intercal)  , \text{ and } \\
    \Omega_t 
    &= \Omega_{t\-1} + y_ty_t^\intercal \! + \! M_{t\-1}^{\intercal}\Lambda_{t\-1}M_{t\-1} \\
    & - \! (\Lambda_{t\-1}M_{t\-1} \! + \! x_t y_t^\intercal)^{\intercal} (\Lambda_{t\-1} \! + \! x_tx_t^\intercal)^{\-1} (\Lambda_{t\-1}M_{t\-1} \! +\! x_t y_t^\intercal)  . \nonumber
\end{align}
\end{lemma}
See the corresponding proof in Appendix~\ref{sec:app:params}. This solution can be expressed as a message passing procedure on a factor graph, enabling distributed computation \cite{loeliger2007factor}. 

The circled numbers in Fig.~\ref{fig:ffg} denote messages passed between factor nodes along edges. Message $\wcircled{1}$ represents the previous posterior belief over coefficients and precision:
\begin{align} \label{eq:message1}
    \overrightarrow{\wcircled{1}} &= p(A,W \given \mathcal{D}_{t\-1}) \\
    &= \mathcal{MNW}(A, W \given M_{t\-1}, \Lambda_{t\-1}^{-1}, \Omega_{t\-1}^{-1}, \nu_{t\-1}) \, .
\end{align}

The sum-product message from the composite MARX likelihood towards its parameters is the likelihood function itself, re-expressible as a probability distribution over $A,W$. 
\begin{lemma} \label{eq:lemma:backwards}
    The message from the composite MARX likelihood \eqref{eq:param:likelihood} towards its parameters is matrix normal Wishart distributed:
    \begin{align} \label{eq:message2}
        \uparrow \wcircled{2} &= p(y_t \given A,W, x_t) \\
        &\propto \mathcal{MNW}(A,W \given \bar{M}_t, \bar{\Lambda}_t^{-1}, \bar{\Omega}^{-1}_t, \bar{\nu}_t) \, .
    \end{align}
    Its parameters are:
    \begin{align}
        \bar{\nu}_t &= 2 - D_x + D_y \, , \qquad \ 
        \bar{\Lambda}_t = x_t x_t^{\intercal} \, , \\
        \bar{M}_t &= (x_t x_t^{\intercal})^{-1} x_t y_t^{\intercal} \, , \qquad 
        \bar{\Omega}_t = 0_{D_y \times D_y} \, .
    \end{align}
\end{lemma}
The proof is in Appendix~\ref{sec:app:backward}.

Message $\wcircled{3}$ is the result of the multiplication of messages $\wcircled{1}$ and $\wcircled{2}$ in the equality node \cite{loeliger2007factor}.  
\begin{lemma} \label{eq:lemma:product}
    Let $p_1,p_2$ be two matrix normal Wishart distributions over the same random variables $A,W$:
    \begin{align}
        p_1(A,W) &= \mathcal{MNW}(A,W \given M_1, \Lambda_1^{-1}, \Omega_1^{-1}, \nu_1) \\
        p_2(A,W) &= \mathcal{MNW}(A,W \given M_2, \Lambda_2^{-1}, \Omega_2^{-1}, \nu_2) \, .
    \end{align}
    Their product is proportional to another matrix normal Wishart distribution:
    \begin{align}
        p_1( A, \! W  ) p_2( A, \! W )  \propto  \mathcal{MNW}(A,\! W | M_3, \Lambda_3^{\-1}, \Omega_3^{\-1}, \nu_3) 
    \end{align}
    whose parameters are combinations of $p_1,p_2$'s parameters,
    \begin{align}
    \nu_3 &= \nu_1 + \nu_2 + D_x - D_y -1 \, , \\ 
    \Lambda_3 &= \Lambda_1 + \Lambda_2 \, , \\
    M_3 &= (\Lambda_1 + \Lambda_2)^{-1} (\Lambda_1 M_1 + \Lambda_2 M_2) \, , \\
    \Omega_3 &= \Omega_1  +  \Omega_2  + M_1^{\intercal} \Lambda_1 M_1  +  M_2^{\intercal} \Lambda_2 M_2 \\
    & - \! (\Lambda_1 M_1 \! + \! \Lambda_2 M_2)^{\intercal} (\Lambda_1 \! + \! \Lambda_2)^{\-1} (\Lambda_1 M_1 \! + \! \Lambda_2 M_2) . \nonumber
\end{align}
\end{lemma}
See Appendix~\ref{sec:app:product} for the proof.
\begin{theorem}
    The equality node's outgoing message is proportional to the exact recursive posterior distribution;
    \begin{align}
        \overrightarrow{\wcircled{3}} = \overrightarrow{\wcircled{1}} \cdot \wcircled{2} \uparrow
        &\propto \mathcal{MNW}(A,W \given M_t, \Lambda_t^{-1}, \Omega_t^{-1}, \nu_t)
    \end{align}
\end{theorem}
\begin{proof} Combining parameters from the messages in \eqref{eq:message1} and \eqref{eq:message2} according the product operation in Lemma~\ref{eq:lemma:product}, yields:
\begin{align}
    \nu_t &= \nu_{t\-1} + 1 \, , \\ 
    \Lambda_t &= \Lambda_{t\-1} + x_t x_t^{\intercal} \, , \\
    M_t &= (\Lambda_{t\-1} + x_t x_t^{\intercal})^{-1}(\Lambda_{t\-1} M_{t\-1} + x_t y_t^{\intercal}) \, , \\
    \Omega_t &= \Omega_{t\-1} + M_{t\-1}^{\intercal} \Lambda_{t\-1} M_{t\-1} + y_t y_t^{\intercal} \\
    &\quad \tm (\Lambda_{t\-1} M_{t\-1} \tp x_t y_t^{\intercal})^{\intercal} (\Lambda_{t\-1} \tp x_t x_t^{\intercal})^{\-1} (\Lambda_{t\-1} M_{t\-1} \tp x_t y_t^{\intercal}) \nonumber
\end{align}
These match the parameter update rules of Lemma~\ref{eq:lemma1}.    
\end{proof}

\subsection{Output prediction} \label{sec:inference:prediction}
\wouter{Consider removing output prediction section entirely (which means experiments will purely be based on parameter space evaluations. This would give room for details on derivations, discussions on distributed computing.}
Predicting future system outputs in our probabilistic model constitutes deriving the posterior predictive distribution, i.e., the marginal distribution of $y_{\tau}$ for $\tau > t$:
\begin{align}
 \downarrow \wcircled{4} &= p(y_{\tau} \given u_{\tau}, \mathcal{D}_t)  \\
&= \int p(y_{\tau} \given A, W, x_{\tau}) p(A, W \given \mathcal{D}_t) \; \mathrm{d}(A, W) \, . \label{eq:postpred}
\end{align}
We can utilize the factorisation structure of the parameter posterior to split this into a marginalization over $A$,
\begin{align}
p(y_{\tau} | W, u_{\tau}, \mathcal{D}_t) \! = \! \int \! p(y_{\tau} | A, W, x_{\tau}) p(A | W, \mathcal{D}_t)  \mathrm{d}A \, ,
\end{align}
and a marginalization over $W$,
\begin{align}
    p(y_{\tau} \given u_{\tau},  \mathcal{D}_t)  = \! \int \! p(y_{\tau} \given W, u_{\tau},  \mathcal{D}_t) p(W \given \mathcal{D}_t) \mathrm{d}W \, .
\end{align}
\begin{theorem}
The marginalization of the composite MARX likelihood function \eqref{eq:param:likelihood} over a matrix normal distribution \eqref{eq:priors:A} yields a multivariate normal distribution:
\begin{align}
    \int  \mathcal{N}\big(y_{\tau} &\given A^{\intercal} x_{\tau}, W^{-1} \big) \mathcal{MN}\big(A \given M_t, \Lambda_t^{-1}, W^{-1} \big) \mathrm{d} A \nonumber \\
    &= \mathcal{N}\big(y_{\tau} \given M_t^{\intercal} x_{\tau}, (\lambda_{\tau} W)^{-1} \big) \, ,
\end{align}
where $\lambda_{\tau} \triangleq (1 + x_{\tau}^{\intercal} \Lambda_t^{-1} x_{\tau} )^{-1}$.
\end{theorem}

The proof is in Appendix~\ref{sec:app:marginalA}. 
The marginalization of a multivariate normal distribution over a Wishart distribution for its precision parameter yields a multivariate location-scale T-distribution \cite[ID: D148]{joram_soch_2024_10495684}:
\begin{align}
    \int \mathcal{N}\big(y_{\tau} &\given M_t^{\intercal} x_{\tau}, (\lambda_{\tau} W)^{-1} \big) \mathcal{W}(W \given \Omega_t^{-1}, \nu_t) \mathrm{d}W \nonumber \\
    &= \mathcal{T}(y_{\tau} \given \mu_{\tau}, \Psi_{\tau}^{-1}, \eta_{\tau}) \, ,
\end{align}
where $\mu_{\tau} \triangleq M_t^{\intercal}x_{\tau}$, 
$\eta_{\tau} \triangleq \nu_t - D_y + 1$, and $\Psi_{\tau} \triangleq \eta_{\tau} \Omega_t^{-1} \lambda_{\tau}$.
The posterior predictive distribution provides a recursive uncertainty estimate of the output predictions, which is valuable for decision-making and adaptive control. 

\section{EXPERIMENTS} \label{sec:experiments}
We evaluate the MARX estimator on two systems: a multivariate autoregressive system (verification) and a double mass-spring-damper system (validation)\footnote{Code: \url{https://github.com/biaslab/ECC2025-MARXEFE}}.

\subsection{Baseline estimator}
In both experiments, we compare against a recursive least-squares (RLS) estimator \cite{hannan1989recursive}. The coefficient point estimate $\hat{A}_t$ and (initial) inverse sample covariance matrix $P_0 = I_{D_x}$ are updated at each timestep via
\begin{align}
    P_t &= P_{t\-1} - P_{t\-1} x_t (1 \tp x_t^{\intercal} P_{t\-1} x_t )^{-1} x_t^{\intercal} P_{t\-1}  \\
    \hat{A}_t &= \hat{A}_{t\-1} + P_{t\-1} x_t (1 \tp x_t^{\intercal} P_{t\-1} x_t )^{-1} (y_t \tm \hat{A}_{t\-1}^{\intercal} x_t)^{\intercal}  \, .
\end{align}
This corresponds to a forgetting factor of $1.0$, meaning all data points are weighted equally.
Predictions are given by $y_\tau = \hat{A}^{\intercal}_t x_{\tau}$.

\subsection{Verification} \label{sec:experiments:verification}
The verification system uses $z_t = x_t$ with $N_y = 2$, $N_u = 3$, and $D_y = D_u = 2$.
True coefficients $\tilde{A}$ are generated using Butterworth filters (20\,Hz cutoff) for self-connections and Gaussian noise (mean $0$, std $0.1$) for cross-connections.
Disturbances follow $e_t \sim \mathcal{N}(0, \tilde{W}^{-1})$ with $\tilde{W} = \begin{bmatrix} 300 & 100 \\ 100 & 200 \end{bmatrix}$.
Training sizes $T_{\text{train}} \in \{2, 4, 6, \ldots, 64\}$ are evaluated over $N_{\text{MC}}=600$ Monte Carlo runs and $T_{\text{test}}=100$ test steps.
We compare an uninformative prior (MARX-UI, $\Lambda_0=1e\tm4 \cdot I_{D_x}$, $\Omega_0=1e\tm5 \cdot I_{D_y}$) and weakly informative prior (MARX-WI, $\Lambda_0=1e\tm1 \cdot I_{D_x}$, $\Omega_0=1e\tm1 \cdot I_{D_x}$); both with $M_0=0_{D_x \times D_y}, \nu_0=D_y+2$.

\begin{figure}[htbp]
    \centering
    \includegraphics[width=\linewidth]{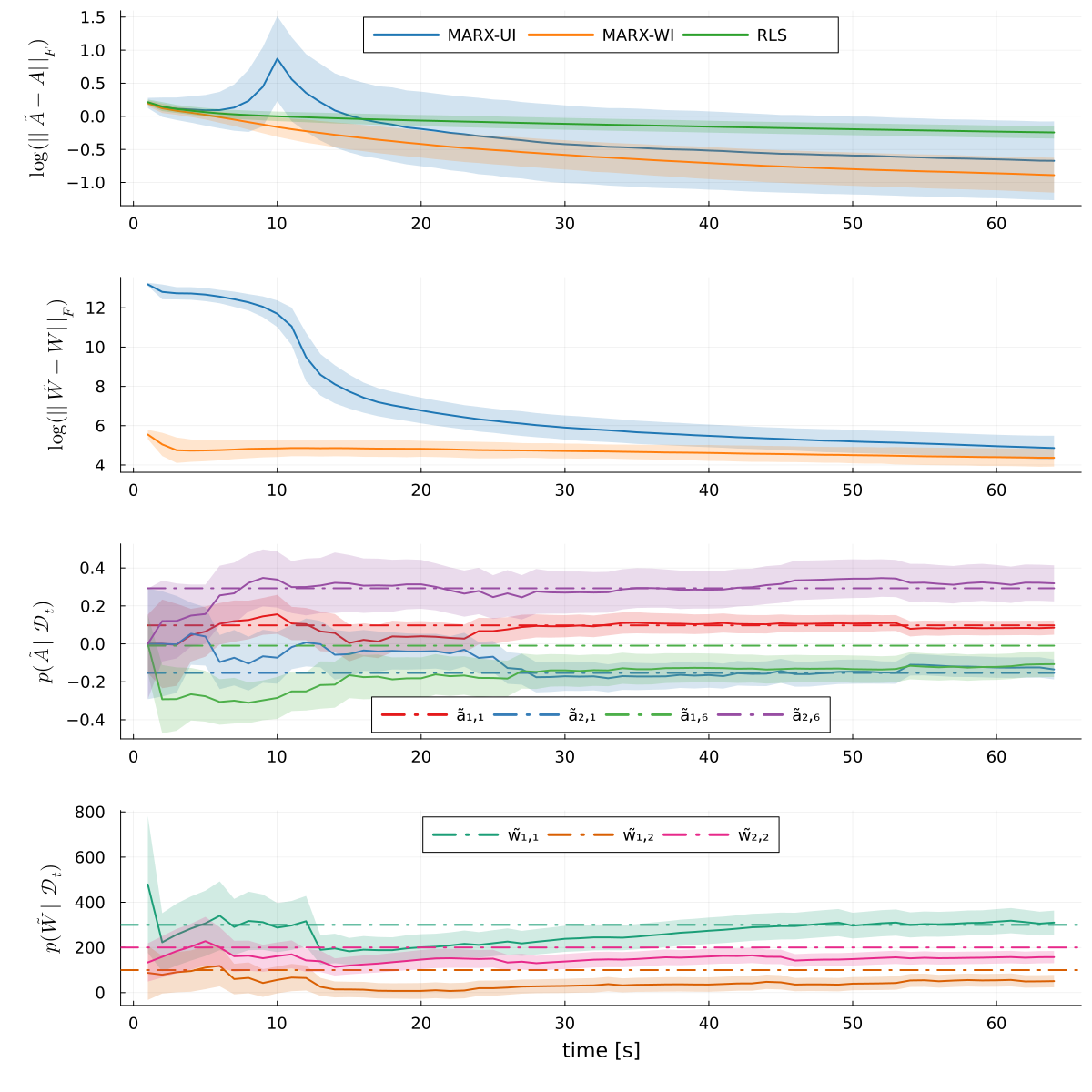} 
    \caption{(\textbf{Top}) Log-scale Frobenius-norm differences ($\tilde{A}$ vs.\ $A$); (\textbf{Second}) $\tilde{W}$ vs $W$. Time series of selected elements of $A$ (\textbf{third}) and $W$ (\textbf{bottom}). Shaded areas denote standard errors; horizontal lines indicate true values.}
    \label{fig:experiments:MARX:param-estim-runs}
\end{figure}

At $T_{train} = 2^6$, Fig.~\ref{fig:experiments:MARX:param-estim-runs} shows that MARX-WI consistently achieves better $A$ estimates than MARX-UI and RLS; MARX-UI outperforms RLS after initial instability.
Unlike RLS, the MARX estimators also estimate $W$.
Since the MARX estimator is probabilistic, it quantifies uncertainty in its estimations of $\tilde{A}$ and $\tilde{W}$ through precision parameters.
The last two subplots in Fig.~\ref{fig:experiments:MARX:param-estim-runs} illustrate the evolution of MARX-WI estimates and their standard deviation ribbons for selected elements of $A$ and $W$, with uncertainty decreasing over time.
All estimators converge to comparable root mean squared errors (RMSE) (with standard errors): $0.284 \pm 5.45e\tm3$ (MARX-WI), $0.289 \pm 5.43e\tm3$ (MARX-UI), and $0.301 \pm 5.50e\tm3$ (RLS).

\subsection{Validation} \label{sec:experiments:validation}
The validation system is a double mass-spring-damper system.
Mass $m_1$ is attached to a fixed base via spring and damper coefficients $k_1$ and $c_1$, and mass $m_2$ is linked to $m_1$ via $k_2$ and $c_2$.
The discrete-time dynamics follow a first-order ordinary differential equation (ODE) $I_t \ddot{z}_t = F(z_t, \dot{z}_t, u_t)$ with diagonal inertia matrix $I_t $ with elements $m_1$ and $m_2$, generalized coordinates $z_t$, first and second time derivatives of $z_t$ ($\dot{z}_t$ and $\ddot{z}_t$), and a generalized force function $F(z_t, \dot{z}_t, u_t) = K z_t + C \dot{z}_t + u_t$ with spring and damping matrices
\begin{align}
K &= \begin{bmatrix}
    \tm (k_1 \tp k_2) &     k_2 \\
                 k_2  & \tm k_2
    \end{bmatrix}, \quad
C = \begin{bmatrix}
    \tm (c_1 \tp c_2) &    c_2 \\
                c_2  & \tm c_2
    \end{bmatrix} \, .
\end{align}
To evolve the system over discrete time steps $\Delta t = 0.05$, we split its equations of motion into a set of two coupled second-order ODEs $z_{\tau} = z_t + \Delta t \dot{z}_t$ and $\dot{z}_{\tau} = \dot{z}_{t} + \Delta t \ddot{z}_t$, solving the system of ODEs with forward Euler.
At $T_{\text{train}}=2^6$, MARX-WI and MARX-UI achieve lower RMSEs and standard error than RLS: $0.048 \pm 1.26e\tm3$ for (MARX-WI), $0.046 \pm 1.26e\tm3$ (MARX-UI), and $0.074 \pm 2.24e\tm3$ (RLS).

\section{DISCUSSION} \label{sec:discussion}

\wouter{todo: paragraph on interpretation of results?}

The modularity of the factor graph approach offers significant practical benefits.
As shown in \cite{loeliger2007factor}, they enable the visual design of sophisticated algorithms by adding, removing or combining known computational blocks. For example, the factor graph in Fig.~\ref{fig:ffg} can be extended to a time-varying MARX by adding state transition factor nodes between the equality nodes over parameters \cite{podusenko2021message}.
In multi-agent robotics, where sensors and actuators are dispersed across platforms, each agent can update its local beliefs via message passing and share only the most informative summaries \cite{damgaard2022study}.
This selective communication can minimize bandwidth while rapidly converging to an accurate global model.
Recent work underscores the value of transmitting informative variational beliefs in multi-agent settings \cite{tedeschini2023message, zhang2024message}, which allows scalable cooperative learning among heterogeneous agents.
Furthermore, distributing computation across nodes opens up promising avenues in federated system identification and multi-robot coordination, especially under privacy or bandwidth constraints \cite{ta2014factor,castaldo2014multi,van2024multi}.
Overall, the factor graph approach fosters robust, scalable, and efficient inference in distributed sensing, control and robotics, as local computations are fused to form a coherent global model.


\section{CONCLUSIONS} \label{sec:conclusions}
We presented a recursive Bayesian estimation procedure for multivariate autoregressive models with exogenous inputs formulated as message passing on a Forney-style factor graph. It produces matrix-variate posterior distributions for the coefficients and noise precision, whose uncertainty is propagated into future system output predictions. 

\begin{appendices}
\section{Parameter estimation} \label{sec:app:params}
\wouter{Consider introducing lemma to avoid duplicate steps.}
\wouter{Use short-hand notation to reduce appendices.}
\begin{proof} \label{proof:param-posterior}
The functional form of the likelihood is 
\begin{align}
p(y_t \given A, W, x_t) 
&\propto \sqrt{|W|} \exp\big(\tm \frac{1}{2}\tr\big[W L_t \big] \big) \, , 
\end{align}
where $L_t \triangleq (y_t - A^{\intercal} x_t)(y_t - A^{\intercal}x_t)^\intercal$.
The prior is
\begin{align}
p(A,\! W | \mathcal{D}_{t\-1}) 
\! \propto  \! \! \sqrt{\! |W|^{\nu_{t\-1} \tp \bar{D}}} \!\exp(\- \frac{1}{2}\tr[W(H_{t\-1} \+ \Omega_{t\-1})]) ,
\end{align}
where $H_{t\-1} \triangleq (A - M_{t\-1})^{\intercal}\Lambda_{t\-1}(A - M_{t\-1})$ and $\bar{D} \triangleq D_x - D_y - 1$.
The posterior is proportional to the prior times likelihood:
\begin{align}
&p(A, W \given \mathcal{D}_t) \propto p(y_t \given A, W, x_t) \; p(A, W \given \mathcal{D}_{t\-1}) \\
 &\propto \! \sqrt{|W|^{\nu_{t\-1} + 1 + \bar{D}}}  
\exp\big(\tm \frac{1}{2}\tr\big[W(L_t \tp H_{t\-1} \tp \Omega_{t\-1})\big] \big) . \label{eq:posterior:conjugate}
\end{align}
We expand the first terms in the exponent and group them:
\begin{align}
L_t + &H_{t\-1} =  y_ty_t^\intercal - y_tx_t^\intercal A - A^{\intercal} x_t y_t^\intercal + A^{\intercal} x_t x_t^{\intercal} A \\
&\ +  A^{\intercal}\Lambda_{t\-1}A \tm A^{\intercal}\Lambda_{t\-1}M_{t\-1} \tm M_{t\-1}^{\intercal}\Lambda_{t\-1}A \tp M_{t\-1}^{\intercal}\Lambda_{t\-1}M_{t\-1} \nonumber \\
&\ = A^{\intercal}(\Lambda_{t\-1} \! + \! x_t x_t^{\intercal}) A \! - \! A^{\intercal}(x_t y_t^{\intercal} \! + \! \Lambda_{t\-1} M_{t\-1}) \label{eq:proofparam1} \\
&\ - (M_{t\-1}^{\intercal} \Lambda_{t\-1} + y_t x_t^{\intercal})A + y_t y_t^{\intercal} + M^{\intercal}_{t\-1} \Lambda_{t\-1} M_{t\-1}  . \nonumber
\end{align}
Let $\Lambda_t \triangleq \Lambda_{t\-1} + x_t x_t^{\intercal}$, $\xi_t \triangleq x_t y_t^{\intercal}  +  \Lambda_{t\-1} M_{t\-1}$ and $M_t \triangleq \Lambda_t^{-1} \xi_t$. Adding and subtracting $\xi_t^{\intercal} \Lambda_t^{-1} \xi_t$ to \eqref{eq:proofparam1} yields:
\begin{align}
    L_t \! +\! H_{t\-1} &=  A^{\intercal}\Lambda_{t} A \! - \! A^{\intercal} \xi_t - \xi_t^{\intercal} A + \xi_t^{\intercal} \Lambda_t^{-1} \xi_t \\
    &\quad - \xi_t^{\intercal} \Lambda_t^{-1} \xi_t + y_t y_t^{\intercal} + M^{\intercal}_{t\-1} \Lambda_{t\-1} M_{t\-1} \nonumber \\
    &= (A - \Lambda_t^{-1} \xi_t)^{\intercal} \Lambda_{t} (A - \Lambda^{-1}_t  \xi_t) \label{eq:proofparam2} \\
    &\quad - M_t^{\intercal} \Lambda_t M_t + y_t y_t^{\intercal} + M^{\intercal}_{t\-1} \Lambda_{t\-1} M_{t\-1} \, . \nonumber
\end{align}
%
Plugging the above into \eqref{eq:posterior:conjugate}, we recognize the functional form of the matrix normal Wishart distribution:
\begin{align}
&\sqrt{|W|^{\nu_t + \bar{D}}} \exp\big(\tm \frac{1}{2} \tr\big[W( (A \tm M_t)^{\intercal} \Lambda_{t} (A \tm M_t) \! + \! \Omega_t)\big] \big)  \nonumber \\
&\qquad \propto \mathcal{MNW}(A, W \given M_t, \Lambda_t^{-1}, \Omega_t^{-1}, \nu_t)
\end{align}
whose parameters are
\begin{align}
    \nu_t
    &= \nu_{t\-1} + 1 \, , \\
    \Lambda_t
    &= \Lambda_{t\-1} + x_tx_t^\intercal \, , \\
    M_t
    &= (\Lambda_{t\-1} + x_t x_t^{\intercal})^{-1} (\Lambda_{t\-1}M_{t\-1}  + x_t y_t^\intercal)  \, , \text{ and } \\
    \Omega_t 
    &= \Omega_{t\-1} + y_ty_t^\intercal + M_{t\-1}^{\intercal} \Lambda_{t\-1}M_{t\-1} - M_t^\intercal\Lambda_tM_t \, .
\end{align}
This concludes the proof.
\end{proof}

\section{Backwards message from likelihood} \label{sec:app:backward}

\begin{proof}
The MARX likelihood function is:
\begin{align}
    p(y_t \given A,W, x_t)  
    &\propto \sqrt{|W|} \exp\big( \tm \frac{1}{2}\tr\big[ W L_t \big] \big) \label{eq:proof:backward1} \, ,
\end{align}
where the completed square is
\begin{align}
    L_t &\triangleq (y_t - A^{\intercal} x_t)(y_t - A^{\intercal} x_t)^{\intercal}  \\
    &= y_t y_t^{\intercal} - A^{\intercal} x_t y_t^{\intercal} - y_t x_t^{\intercal} A + A^{\intercal} x_t x_t^{\intercal} A \, .
\end{align}
Let $\bar{\Lambda}_t \triangleq x_t x_t^{\intercal}$, $\bar{\xi}_t \triangleq x_t y_t^{\intercal}$ and $\bar{M}_t = \bar{\Lambda}_t^{-1} \bar{\xi}_t$. Then adding and subtracting $\bar{\xi}_t^{\intercal} \bar{\Lambda}_t \bar{\xi}_t$ allows us to rewrite the square in terms of $A$: 
\begin{align}
    &L_t + \bar{\xi}_t^{\intercal} \bar{\Lambda}_{t}^{-1} \bar{\xi}_t - \bar{\xi}_t^{\intercal} \bar{\Lambda}_{t}^{-1} \bar{\xi}_t  \\
    &\quad = y_t y_t^{\intercal} + (A - \bar{M}_t)^{\intercal} \bar{\Lambda}_t (A - \bar{M}_t) - \bar{\xi}_t^{\intercal} \bar{\Lambda}_{t}^{-1} \bar{\xi}_t \, . \nonumber
\end{align}
The two remaining terms cancel:
\begin{align}
    y_t y_t^{\intercal} - \bar{\xi}_t^{\intercal} \bar{\Lambda}_{t}^{-1} \bar{\xi}_t &= y_t y_t^{\intercal} - y_t x_t^{\intercal} (x_t x_t^{\intercal})^{-1} x_t y_t^{\intercal} \\ 
    &= y_t y_t^{\intercal} - y_t I y_t^{\intercal} = 0_{D_y \times D_y} \, .
\end{align}
If we define $\bar{\nu}_t \triangleq 1 - \bar{D}$ for $\bar{D} = D_x + D_y + 1$ and $\bar{\Omega}_t \triangleq 0_{D_y \times D_y}$, then we may recognize the functional form of a matrix normal Wishart in \eqref{eq:proof:backward1};
\begin{align}
    &p(y_t | A,W, x_t) \\
    &\propto \! \sqrt{|W|^{ \bar{\nu}_t + \bar{D}}} \! \exp\big( \! \tm \frac{1}{2} \tr\big[ W ((A \! - \! \bar{M}_t)^{\intercal} \bar{\Lambda}_t (A \! - \! \bar{M}_t) \! + \! \bar{\Omega}_t ) \big] \big) \nonumber \\
    &\propto \mathcal{MNW}(A,W \given \bar{M}_t, \bar{\Lambda}_t^{-1}, \bar{\Omega}_t^{-1}, \bar{\nu}_t) \, .
\end{align}
This concludes the proof.
\end{proof}

\section{Product of matrix normal Wishart distributions} \label{sec:app:product}
\begin{proof}
Let $p_1,p_2$ be two matrix normal Wishart distributions over the same random variables $A,W$:
\begin{align}
    p_1(A,W) &= \mathcal{MNW}(A,W \given M_1, \Lambda_1^{-1}, \Omega_1^{-1}, \nu_1) \\    p_2(A,W) &= \mathcal{MNW}(A,W \given M_2, \Lambda_2^{-1}, \Omega_2^{-1}, \nu_2) \, .
\end{align}
Their product is proportional to:
\begin{align}
    p_1(A,&W) \, p_2(A,W) \nonumber \\
    &\propto \ \sqrt{|W|^{\nu_1 + \bar{D}}}  \exp\big(-\frac{1}{2}\tr\big[W L_1 \big] \big) \\
    &\quad \ \sqrt{|W|^{\nu_2 + \bar{D}}} \exp\big(-\frac{1}{2}\tr\big[W L_2 \big] \big) \nonumber \\
    &= \sqrt{|W|^{\nu_3 + \bar{D}}}  \exp\big(-\frac{1}{2}\tr\big[W (L_1 + L_2) \big] \big)
\end{align}
for $\bar{D} \triangleq D_x - D_y -1$, $\nu_3 \triangleq \nu_1 + \nu_2 + D_x - D_y - 1$ and $L_i \triangleq (A \! - \! M_i)^{\intercal} \Lambda_i (A \! - \! M_i) + \Omega_i$. The sum of $L_i$ is:
\begin{align}
    L_1 + L_2 &= A^\intercal (\Lambda_1 + \Lambda_2) A \label{eq:proof:product1} \\
    &\quad - A^{\intercal} (\Lambda_1 M_1 + \Lambda_2 M_2)  - (M_1^{\intercal} \Lambda_1 + M_2^{\intercal} \Lambda_2) A \nonumber \\
    &\quad + M_1^{\intercal} \Lambda_1 M_1 + M_2^{\intercal} \Lambda_2 M_2  + \Omega_1 + \Omega_2 \nonumber \, .
\end{align}
Let $\Lambda_3 \triangleq \Lambda_1 + \Lambda_2$ and $\xi_3 \triangleq \Lambda_1 M_1 + \Lambda_2 M_2$. Then:
\begin{align}
    &(A - \Lambda_3^{-1} \xi_3)^{\intercal} \Lambda_3 (A - \Lambda_3^{-1} \xi_3) = \nonumber \\
    &\qquad \quad A^{\intercal} \Lambda_3 A - A^{\intercal} \xi_3 -  \xi_3^{\intercal} A + \xi_3^{\intercal} \Lambda_3^{-1} \xi_3 \, .
\end{align}
Using $M_3 \triangleq \Lambda_3^{-1} \xi_3$, \eqref{eq:proof:product1} can be written as:
\begin{align}
    &p_1(A,W) p_2(A,W)  \label{eq:proof:product2} \\
    &\propto \sqrt{|W|^{\nu_3 + \bar{D}}} \exp\big(\tm \frac{1}{2}\tr\big[W \big( (A - M_3)^{\intercal} \Lambda_3 (A - M_3) \nonumber \\
    &\qquad - \xi_3^{\intercal} \Lambda_3^{-1} \xi_3 \! + \! M_1^{\intercal} \Lambda_1 M_1 \! + \! M_2^{\intercal} \Lambda_2 M_2  \! + \! \Omega_1 \! + \! \Omega_2\big) \big] \big) \nonumber
\end{align}
Note that $\xi_3^{\intercal} \Lambda_3^{-1} \xi_3 = \xi_3^{\intercal} \Lambda_3^{-1} \Lambda_3 \Lambda_3^{-1} \xi_3 = M_3^{\intercal} \Lambda_3 M_3$. Let
\begin{align}
    \Omega_3 \triangleq \Omega_1  +  \Omega_2  + \! M_1^{\intercal} \Lambda_1 M_1 \! + \! M_2^{\intercal} \Lambda_2 M_2 - \! M_3^{\intercal} \Lambda_3 M_3 .
\end{align}
Then \eqref{eq:proof:product2} may be recognized as an unnormalized matrix normal Wishart:
\begin{align}
    &\sqrt{|W|^{\nu_3 + \bar{D}}} \exp\big(\tm \frac{1}{2}\tr\big[W \big((A \tm M_3)^{\intercal} \Lambda_3 (A \tm M_3) \! + \! \Omega_3 \big) \big] \big) \nonumber \\
    &\qquad \propto \mathcal{MNW}\big(A,W \given M_3, \Lambda_3^{-1}, \Omega_3^{-1}, \nu_3 \big) \, . \label{eq:proof:product3}
\end{align}
As such, the product of two matrix normal Wishart distributions is proportional to another matrix normal Wishart distribution.
\end{proof}

\section{Marginalization over matrix normal dist.} \label{sec:app:marginalA}
\begin{proof} \label{proof:prediction}
%
Let $L_{\smallbullet} \triangleq (y_{\smallbullet} - A^{\intercal} x_{\smallbullet})(y_{\smallbullet} - A^{\intercal} x_{\smallbullet})^\intercal$ and $H_{\smallbullet} \triangleq (A - M_{\smallbullet})^{\intercal}\Lambda_{\smallbullet}(A - M_{\smallbullet})$ where the subscript $_{\smallbullet}$ indicates a time index.
The marginalization over $A$ is:
\begin{align}
&p(y_\tau | u_{\tau}, W; \mathcal{D}_t) \! = \! \int p( y_{\tau} | A, W, x_{\tau}) p(A | W, \mathcal{D}_t) \; \mathrm{d}A \\
 &= \! \int \! \mathcal{N}\big(y_{\tau} | A^{\intercal} x_{\tau}, W^{\-1} \big) \mathcal{MN}\big(A | M_t, \Lambda_t^{\-1}, W^{\-1} \big) \mathrm{d} A  \\
&= \sqrt{ (2\pi)^{- D_y (D_{x} + 1)} |W|^{D_x + 1} |\Lambda_t|^{D_y} } \nonumber \\
&\qquad \qquad \int \exp \big(  \tm  \frac{1}{2}\tr\big[ W (L_{\tau}  +  H_t) \big] \big) \mathrm{d}A . \label{eq:proof:marginalA1}
\end{align}
Expanding $L_{\tau}$ and $H_t$ and adding them yields:
\begin{align}
    L_{\tau} \tp H_{t} 
    &=  y_{\tau} y_{\tau}^{\intercal} + M_t^{\intercal} \Lambda_t M_t + A^{\intercal}(\Lambda_t \! + \! x_{\tau} x_{\tau}^\intercal) A  \\
    &\quad - A^{\intercal} (\Lambda_t M_t \! + \! x_{\tau} y_{\tau}^\intercal) - (\Lambda_t M_t \! + \! x_{\tau} y_{\tau}^\intercal)^{\intercal} A .  \nonumber
\end{align}
Let $\Lambda_{\tau} \triangleq \Lambda_t + x_{\tau} x_{\tau}^{\intercal}$, $\xi_{\tau} \triangleq \Lambda_t M_t + x_{\tau} y_{\tau}^{\intercal}$ and $M_{\tau} \triangleq \Lambda_{\tau}^{-1} \xi_{\tau}$. Completing the square gives:
\begin{align}
    L_{\tau} + H_t &=  H_\tau - M_{\tau}^{\intercal} \Lambda_{\tau} M_{\tau} + y_{\tau} y_{\tau}^{\intercal} + M_t^{\intercal} \Lambda_t M_t \, . 
\end{align}
Plugging this result into the integral in \eqref{eq:proof:marginalA1} gives:
\begin{align}
& \! \int \! \! \exp\big(\tm \frac{1}{2}\tr\big[W(L_{\tau} \! \tp \! H_t) \big] \big) \mathrm{d}A \! = \! \! \int \! \! \exp\big(\! -\! \frac{1}{2}\tr\big[W H_\tau \big]\big) \mathrm{d}A  \nonumber \\
&\quad \cdot \exp\big(\!-\!\frac{1}{2}\tr\big[W(y_{\tau} y_{\tau}^\intercal \! + \! M_t^\intercal \Lambda_t M_t \! - \! M_{\tau}^\intercal\Lambda_{\tau} M_{\tau}) \big] \big)  .
\end{align}
We can recognize the integrand as the functional form of a matrix normal distribution. Thus, the integral evaluates to its inverse normalization factor: 
\begin{align}
& \int \exp\big(-\frac{1}{2}\tr\big[W H_{\tau} \big] \big) \; dA = \sqrt{\frac{(2\pi)^{D_y D_x}}{|W|^{D_x} |\Lambda_{\tau}|^{D_y}}}  \, .
\end{align}
Using this result, the marginalization over $A$ becomes:
\begin{align}
& p(y_\tau \given  u_{\tau}, W; \mathcal{D}_t)
 =  \sqrt{(2\pi)^{-D_y} |\Lambda_t|^{D_y} |\Lambda_{\tau}|^{-D_y} |W|} \\
&\qquad \cdot \exp\big(\! \tm  \frac{1}{2}\tr\big[W(y_{\tau}y_{\tau}^\intercal \tp M_t^\intercal \Lambda_t M_t \tm M_{\tau}^\intercal\Lambda_{\tau}M_{\tau})\big]\big) . \nonumber 
\end{align}
Note that, under the matrix determinant lemma,
\begin{align}
    |\Lambda_{\tau}| &= | \Lambda_{t} + x_{\tau} x_{\tau}^{\intercal} | =  |\Lambda_{t}|(1 + x_{\tau}^{\intercal} \Lambda_t^{-1} x_{\tau}  ) \, ,
\end{align}
which implies that the product of determinants is
\begin{align}
    |\Lambda_t|^{D_y} |\Lambda_{\tau}|^{-D_y} &= \big(1 + x_{\tau}^{\intercal} \Lambda_t^{-1} x_{\tau} \big)^{-D_y} \, .
\end{align}
Let $\lambda_{\tau} \triangleq (1 + x_{\tau}^{\intercal} \Lambda_t^{-1} x_{\tau} )^{-1}$. As $W$ is $D_y$-dimensional, $|W| \lambda_{\tau}^{D_y} = |W \lambda_{\tau}|$.
Furthermore, note that
\begin{align}
    M_{\tau}^{\intercal} \Lambda_{\tau} M_{\tau} &= M_t^{\intercal} \Lambda_{t}(x_{\tau} x_{\tau}^{\intercal} + \Lambda_{t} )^{-1} \Lambda_{t} M_t  \\
    &\quad + y_{\tau}x_{\tau}^{\intercal} (x_{\tau} x_{\tau}^{\intercal} + \Lambda_{t} )^{-1} \Lambda_{t} M_t \nonumber \\
    &\quad  + M_t^{\intercal} \Lambda_{t}(x_{\tau} x_{\tau}^{\intercal} + \Lambda_{t} )^{-1} x_{\tau} y_{\tau}^{\intercal} \nonumber \\
    & \quad + y_{\tau} x_{\tau}^{\intercal}  (x_{\tau} x_{\tau}^{\intercal} + \Lambda_{t} )^{-1} x_{\tau} y_{\tau}^{\intercal} \nonumber .
\end{align}
Combining this with the other terms in the trace gives:
\begin{align}
    y_{\tau}y_{\tau}^\intercal &+ M_t^\intercal\Lambda_tM_t - M_{\tau}^\intercal\Lambda_{\tau}M_{\tau} \\
    &= M_t^\intercal \Lambda_{t} \big(I - (x_{\tau} x_{\tau}^{\intercal} + \Lambda_{t} )^{-1} \Lambda_{t} \big) M_t \nonumber \\
    &\quad - y_{\tau}x_{\tau}^{\intercal} (x_{\tau} x_{\tau}^{\intercal} + \Lambda_{t} )^{-1} \Lambda_{t} M_t \nonumber \\
    & \quad - M_t^\intercal \Lambda_{t}(x_{\tau} x_{\tau}^{\intercal} + \Lambda_{t} )^{-1} x_{\tau} y_{\tau}^{\intercal} \nonumber \\
    &\quad + y_{\tau} \big(1 - x_{\tau}^{\intercal}  (x_{\tau} x_{\tau}^{\intercal} + \Lambda_{t} )^{-1} x_{\tau} \big) y_{\tau}^{\intercal} \, . \nonumber
\end{align}
Applying the Sherman-Morrison formula and re-arranging terms yields the following simplifications:
\begin{align}
    &\big(1 \tm x_{\tau}^{\intercal}  (x_{\tau} x_{\tau}^{\intercal} \tp \Lambda_{t} )^{-1} x_{\tau} \big) = \lambda_{\tau}  \\
    &I \tm (x_{\tau} x_{\tau}^{\intercal} \tp \Lambda_{t} )^{-1} \Lambda_{t} = \lambda_{\tau} \Lambda_t^{-1} x_{\tau} x_{\tau}^{\intercal} \Lambda_t^{-1} \\
    &\Lambda_{t}(x_{\tau} x_{\tau}^{\intercal} \tp \Lambda_{t} )^{-1} x_{\tau} = x_{\tau} \lambda_{\tau} \, .
\end{align}
Using these three simplifications, we have:
\begin{align}
    \tr\Big[W&\big(y_{\tau}y_{\tau}^\intercal + M_t^\intercal \Lambda_t M_t - M_{\tau}^\intercal\Lambda_{\tau}M_{\tau}\big)\Big] \nonumber \\
    %
    %
    %
    &\qquad = \! \big(y_{\tau} \tm M_t^{\intercal} x_{\tau} \big)^{\intercal} W \lambda_{\tau} \big(y_{\tau} \tm M_t^{\intercal} x_{\tau} \big)  . \label{eq:proof:marginalA3}
\end{align}
%
Plugging \eqref{eq:proof:marginalA3} into \eqref{eq:proof:marginalA1} yields
\begin{align}
 &p(y_\tau \given  u_{\tau}, W; \mathcal{D}_t)  \\
&= \sqrt{\frac{|W \lambda_{\tau}|}{(2\pi)^{D_y}}} \exp\big(\tm \frac{\lambda_{\tau}}{2} (y_{\tau} \tm M_t^{\intercal} x_{\tau} )^{\intercal} W (y_{\tau} \tm M_t^{\intercal} x_{\tau}) \big) \nonumber \\
&= \mathcal{N}\big(y_{\tau} \given M_{t}^{\intercal} x_{\tau}, (W\lambda_{\tau})^{-1}\big) \, .
\end{align}
This concludes the proof.
\end{proof}

\end{appendices}


\bibliographystyle{ieeetr}
\bibliography{references}

\end{document}